\documentclass[a4paper,11pt]{article}
\usepackage[utf8x]{inputenc}
\usepackage{a4wide}
\usepackage{amssymb}
\usepackage{amsmath}
\usepackage{amsthm}
\usepackage{latexsym}
\usepackage[mathscr]{euscript}
\usepackage{hyperref}

\theoremstyle{definition}
\newtheorem{definition}{Definition}[section]

\theoremstyle{plain}
\newtheorem{Theorem}[definition]{Theorem}
\newtheorem{Proposition}[definition]{Proposition}
\newtheorem{Lemma}[definition]{Lemma}

\theoremstyle{remark}
\newtheorem{remark}[definition]{Remark}

\newcommand{\R}{\mathbb R}  
\newcommand{\N}{\mathbb N}

\newcommand{\eps}{\varepsilon}

\usepackage[color]{changebar}
\cbcolor{blue}
\newcommand{\Ric}{\mathrm{Ric}}
\newcommand{\gec}{{\check g_\eps}}

\newcommand{\comp}{\Subset}

\newcommand{\sse}{\subseteq}
\newcommand{\conv}{\mathbf{k}}

\newcommand{\cT}{{\mathcal{T}}}

\title{The Penrose singularity theorem in regularity $C^{1,1}$}

\author{Michael Kunzinger\footnote{University of Vienna, Faculty of Mathematics, 
michael.kunzinger@univie.ac.at, roland.steinbauer@univie.ac.at}, 
Roland Steinbauer\footnotemark[\value{footnote}], 
James A.\ Vickers\footnote{University of Southampton, School of Mathematics, J.A.Vickers@maths.soton.ac.uk}}

\begin{document}

\date{\today}


\maketitle

\begin{abstract}
We extend the validity of the Penrose singularity theorem to spacetime metrics of regularity
$C^{1,1}$. The proof is based on regularisation techniques, combined with
recent results in low regularity causality theory.

\vskip 1em

\noindent
\emph{Keywords:} Singularity theorems, low regularity, regularisation,
causality theory
\medskip

\noindent 
\emph{MSC2010:} 83C75, 
                53B30  

\end{abstract}

\section{Introduction}

In 1965 Roger Penrose published his seminal paper \cite{Pen} which established
the first of the modern singularity theorems. In this paper Penrose
introduced the notion of a trapped surface $\cT $,
which he defined as ``a closed spacelike, two-surface with
the property that the two systems of null geodesics which meet $\cT $
orthogonally converge locally in future directions at $\cT $''. He then
showed that if the spacetime $M$ possesses both a closed trapped
surface and a non-compact Cauchy surface then provided the local
energy density is always positive (so that via Einstein's equations
the Ricci tensor satisfies the null convergence condition) the
spacetime cannot be future null complete. The Penrose paper established
for the first time that the gravitational singularity found in the
Schwarzschild solution was not a result of the high degree of
symmetry but that provided the gravitational collapse qualitatively
resembles the spherically symmetric case then (subject to the above
conditions) deviations from spherical symmetry cannot
prevent the formation of a gravitational singularity.

Penrose's paper was not only the first to define the notion of a
trapped surface but it also introduced the idea of using geodesic
incompleteness to give a mathematical characterisation of a singular
spacetime. The 1965 paper had immediate impact and inspired a series
of papers by Hawking, Penrose, Ellis, Geroch and others which led to
the development of modern singularity theorems (see the recent review
paper \cite{SenGar} for details). Despite the great power of these
theorems they follow Penrose in defining singularities in terms of
geodesic incompleteness and as a result say little about the nature of
the singularity. In particular there is nothing in the original
theorems to say
that the gravitational forces become unbounded at the singularity\footnote{
{{See however results on the extendability 
of incomplete spacetimes under suitable curvature conditions, e.g.\ \cite{Cl82, Clarke,Racz,Thorpe},}{ which indicate that such spacetimes cannot be maximal unless the curvature blows up.}}}.
Furthermore the statement and proofs of the various singularity
theorems assume that the metric is at least $C^2$ and Senovilla in
\cite[Sec.\ 6.1]{Seno1} highlights the places where this assumption is
explicitly used. Thus the singularities predicted by the singularity
theorems could in principle be physically innocuous and simply be a
result of the differentiability of the metric dropping below $C^2$. As
emphasised by a number of authors (see e.g.\ \cite{Clarke,MS,Seno1})
the requirement of $C^2$-differentiability is significantly stronger
than one would want since it fails to hold in a number of physically
reasonable situations. In particular it fails across an interface
(such as the surface of a star) where there is a jump in the energy
density which, via the field equations, corresponds to the metric 
being of regularity $C^{1,1}$ (also denoted by $C^{2-}$, the first
derivatives of the metric being locally Lipschitz continuous). For more
details see e.g.\ \cite[Sec.\ 6.1]{Seno1}.
Furthermore from the point of view of the singularity
theorems themselves the natural differentiability class for the
metric again is $C^{1,1}$ as this is the minimal condition which 
ensures existence and uniqueness of geodesics.  
Since the connection of a $C^{1,1}$-metric is locally Lipschitz,
Rademacher's theorem implies that it is differentiable almost
everywhere so that the (Ricci) curvature exists almost everywhere and
is locally bounded. Any further lowering of the differentiability
would result in a loss of uniqueness of causal geodesics\footnote{
In fact, uniqueness is lost for metrics
of local H\"older regularity class $C^{1,\alpha}$ ($\alpha<1$), see \cite{HW}.} (and hence of
the worldlines of observers) and generically in unbounded curvature\footnote{
While the curvature can be stably defined as a distribution even for metrics
of local Sobolev regularity $W^{1,2}\cap L^\infty$ (\cite{GT}) the curvature  will in general 
not be in $L^\infty$ unless the metric is $C^{1,1}=W^{2,\infty}$.},
both of which correspond more closely to our physical expectations of
a gravitational singularity than in the $C^2$-case.

The singularity theorems involve an interplay between results in
differential geometry and causality theory and it is only recently
that the key elements of $C^{1,1}$-causality have been established. In
particular it was only in \cite[Th.\ 1.11]{M} and in \cite[Th.\ 2.1]{KSS} 
that the exponential map was shown to be a bi-Lipschitz homeomorphism, a key
result needed to derive many standard results in causality theory. 
Building on the regularisation results of
\cite{CG,KSSV} and combining them with recent advances in
causality theory \cite{Chrusciel_causality, CG, M, KSSV} the present
authors in \cite{hawkingc11}  gave a detailed proof of the Hawking singularity theorem for
$C^{1,1}$-metrics by following the basic strategy outlined in
\cite[Sec.\ 8.4]{HE}. In the present paper we establish the Penrose
singularity theorem for a $C^{1,1}$-metric. To be precise we prove
the following result:

\begin{Theorem}\label{penrose} Let $(M,g)$ be an $n$-dimensional $C^{1,1}$-spacetime. Assume
\begin{itemize}
\item[(i)] For any Lipschitz-continuous local null vector field $X$,
$\Ric(X,X)\ge 0$.
\item[(ii)] $M$ possesses a non-compact Cauchy-hypersurface $S$.
\item[(iii)] There exists a compact achronal spacelike submanifold $\cT $
in $M$ of codimension $2$ with past-pointing timelike mean curvature vector field $H$.
\end{itemize}
Then $M$ is not future null complete.
\end{Theorem}
For the definition of a $C^{1,1}$-spacetime, see below.

\begin{remark}\label{rem1.2}\ 
\begin{itemize} 
\item[(a)] 
As explained above the Ricci-tensor, $\Ric$, of a $C^{1,1}$-metric is an (almost everywhere defined)
$L^\infty_{\mbox{\scriptsize loc}}$-tensor field. Condition (i) in Theorem \ref{penrose} is adapted
to this situation and reduces to the usual pointwise condition for metrics
of regularity $C^2$. In fact, any null vector can be extended (by parallel transport)
to a local null vector field that is $C^1$ if the metric is $C^2$ and 
locally Lipschitz if $g$ is $C^{1,1}$ (cf.\ also the proof of Lemma \ref{approxlemma} below).
The assumption in (i) then means that the $L^\infty_{\mbox{\scriptsize loc}}$-function 
$\Ric(X,X)$ is non-negative almost everywhere.
Since being a null vector field is not an `open' condition 
(unlike the case of timelike vector fields as in Hawking's singularity theorem,
see \cite[Rem.\ 1.2]{hawkingc11}),
it will in general not be possible to extend a given null vector to a {\em smooth} local null
vector field. 
\item[(b)] Concerning condition (iii), our conventions are as follows
  (cf.\ \cite{ON83}): we define the mean curvature field as
  $H_p=\frac{1}{n-2}\sum_{i=1}^{n-2}\text{II}(e_i,e_i)$ where
  $\{e_i\}$ is any orthonormal basis of $T_p\cT $ and 
the second fundamental form tensor is given by 
  $\text{II}(V,W)=\text{nor}\nabla_V W$ where $\text{nor}$ denotes 
the projection orthogonal to $T_p\cT$. 
Also the condition on $H$ in (iii) is equivalent to the
  convergence $\conv(v):=g(H,v)$ being strictly positive for all
  future pointing null vectors normal to $\cT $ and with our conventions is therefore
  equivalent to the Penrose trapped surface
  condition. 
\end{itemize}
\end{remark}

The key idea behind Penrose's proof of the $C^2$-theorem is to look at
the properties of the boundary of the future of the trapped surface
$\cT $. The boundary $\partial J^+(\cT )$
is generated by null geodesics but Raychaudhury's
equation and the initial trapped surface condition together with the null
convergence condition result in there being a focal point along every
geodesic. This fact together with the assumption of null geodesic
completeness may be used to show that $\partial J^+(\cT )$ is compact. On
the other hand one may use the existence of the Cauchy surface $S$
together with some basic causality theory to construct a homeomorphism
between $\partial J^+(\cT )$ and $S$. This is not possible if $S$ is not
compact so that there must be a contradiction between the four
assumptions.

In our proof of the theorem for the $C^{1,1}$-case we need to further
extend the methods of \cite{CG, KSS, KSSV,hawkingc11} and approximate $g$ by a
smooth family of Lorentzian metrics $\hat g_\eps$ which have strictly wider
lightcones than $g$ and which are themselves globally hyperbolic.  We then show
that by choosing $\eps$ sufficiently small the associated Ricci
tensor, $\Ric_\eps$, violates the null convergence condition by an
arbitrarily small amount, which allows us to establish the compactness
of $\partial J_\eps^+(\cT )=E_\eps^+(\cT )$ under the assumption of 
null geodesic completeness. We then use the global
hyperbolicity of the $\hat g_\eps$ together with the fact that $S$ is
a Cauchy surface for $g$ to show that $E_\eps^+(\cT )$ is homeomorphic to
$S$, which leads to a contradiction with the non-compactness of $S$. 
Finally, in Theorem \ref{penrose_alt} we show that if $M$ is 
future null complete and the assumption that $S$ be non-compact is dropped
in (ii) then $E^+(\cT )$ is a compact Cauchy-hypersurface in
$M$. A main difficulty in these proofs, as compared to the case of 
Hawking's singularity theorem in \cite{hawkingc11} lies in the
fact that curvature conditions on null vectors are less suitable
for approximation arguments (cf.\ Lemma \ref{approxlemma} below)
than conditions on timelike vectors (`timelike' being an `open' condition,
as opposed to `null').

\medskip

In the remainder of this section we fix key notions to be used throughout this
paper, cf.\ also \cite{hawkingc11}. We assume all
manifolds to be of class $C^\infty$ and connected (as well as Hausdorff and second countable), and
only lower the regularity of the metric. By a  $C^{1,1}$- 
(resp.\ $C^k$-, $k\in \N_0$) spacetime $(M,g)$, we mean a smooth manifold $M$
of dimension $n$ endowed with a Lorentzian metric $g$ of
signature $(-+\dots+)$ possessing locally Lipschitz continuous first
derivatives (resp.\ of class $C^k$) and with a time orientation given by a continuous timelike
vector field.   

If $K$ is a compact set in $M$ we write $K\comp M$.
Following \cite{ON83}, we define the curvature tensor  by
$R(X,Y)Z=\nabla_{[X,Y]}Z - [\nabla_X,\nabla_Y]Z$ and the Ricci
tensor by $R_{ab}=R^c{}_{abc}$. Since both of these conventions differ by a sign from
those of \cite{HE}, the respective definitions of Ricci curvature agree. 
Note also that our definition of the convergence 
$\conv$ follows \cite{ON83} and differs by a sign from that used by some other authors. 

Our notation for causal structures will basically follow \cite{ON83},
although as in \cite{Chrusciel_causality,KSSV} we base all
causality notions on locally Lipschitz curves.  Any
locally Lipschitz curve $c$ is differentiable almost everywhere with locally bounded
velocity.  We call $c$ timelike, causal, spacelike or null, if $c'(t)$ has the
corresponding property almost everywhere.  Based on these notions we
define the relative chronological future $I^+(A,U)$ and causal future
$J^+(A,U)$ of a set $A\subseteq M$ relative to $U\subseteq M$ literally as 
in the smooth case (see \cite[Def.\ 3.1]{KSSV} \cite[2.4]{Chrusciel_causality}).
The future horismos of $A$ is defined as $E^+(A,U)=J^+(A,U)\setminus I^+(A,U)$.
As was shown in \cite[Th.\ 7]{M}, \cite[Cor.\ 3.1]{KSSV}, 
our definitions coincide with the ones based on smooth curves.

A Cauchy hypersurface is a
subset $S$ of $M$ which every inextendible timelike curve intersects
exactly once, see \cite[Def.\ 14.28]{ON83}.  In the smooth case,
for spacelike hypersurfaces this definition of a Cauchy hypersurface 
is equivalent to the one in \cite{HE}, and this remains true in the $C^{1,1}$-case \cite[Prop.\  A.31]{hawkingc11}.
A $C^{1,1}$-spacetime $(M,g)$ is called globally hyperbolic if it is strongly causal
and any causal diamond $J(p,q) = J^+(p)\cap J^-(q)$ is compact. 
It follows from \cite[Lem.\ A.20, Th.\ A.22]{hawkingc11} that $M$ is globally hyperbolic if it
possesses a Cauchy-hypersurface.

We will write
$\exp_p$ for the exponential map of the metric $g$ at $p$, and $\exp_p^{g_\eps}$  for
the one corresponding to the metric $g_\eps$.
For a semi-Riemannian submanifold $S$ of $M$ we denote by $(N(S), \pi)$ its normal
bundle. By  \cite[Th.\ 13]{M}, $N(S)$ is a Lipschitz bundle. 

\section{Approximation results} 

In this section we extend the approximation results of
\cite{hawkingc11} to deal with the fact that we need to be able to
approximate a globally hyperbolic $C^{1,1}$-metric by a smooth family
of globally hyperbolic metrics. In addition we require a more delicate
estimate for the Ricci curvature than that given in \cite[Lemma
3.2]{hawkingc11} due to the fact that the Penrose singularity theorem
makes use of the null convergence condition for the Ricci tensor
rather than the timelike convergence condition used in the Hawking
theorem.

We start by recalling from  \cite[Sec.\ 3.8.2]{ladder}, \cite[Sec.\ 1.2]{CG} 
that for two Lorentzian metrics $g_1$,
$g_2$, we say that $g_2$ has \emph{strictly wider light cones} than $g_1$, denoted by 
\begin{equation}
 g_1\prec g_2, \text{ if for any tangent vector } X\not=0,\ g_1(X,X)\le 0 \text{ implies that } g_2(X,X)<0.
\end{equation}
Thus any $g_1$-causal vector is $g_2$-timelike.
The key result now is \cite[Prop.\ 1.2]{CG}, which we give here in the slightly refined
version of  \cite[Prop.\ 2.5]{KSSV}. Note that the smoothness of the approximating net with 
respect to $\eps$ and $p$ is vital in Proposition \ref{CGrefined} below.

\begin{Proposition}\label{CGapprox} Let $(M,g)$ be a $C^0$-spacetime 
and let $h$ be some smooth
background Riemannian metric on $M$. Then for any $\eps>0$, there exist smooth
Lorentzian metrics $\check g_\eps$ and $\hat g_\eps$ on $M$ such that $\check g_\eps
\prec g \prec \hat g_\eps$ and $d_h(\check g_\eps,g) + d_h(\hat g_\eps,g)<\eps$,
where  
\begin{equation}\label{CGdh}
d_h(g_1,g_2) := \sup_{p\in M,0\not=X,Y\in T_pM} \frac{|g_1(X,Y)-g_2(X,Y)|}{\|X\|_h
\|Y\|_h}.
\end{equation}
Moreover, $\hat g_\eps(p)$ and $\check g_\eps(p)$ depend smoothly on $(\eps,p)\in \R^+\times M$, and if
$g\in C^{1,1}$ then letting $g_\eps$ be either $\check g_\eps$ or $\hat g_\eps$,
we additionally have 
\begin{itemize}
 \item[(i)] $g_\eps$ converges to $g$ in the $C^1$-topology as $\eps\to 0$, and
 \item[(ii)] the second derivatives of $g_\eps$ are bounded, uniformly in $\eps$, on compact sets.
 \end{itemize}
\end{Proposition}
\begin{remark}\label{ghstab}
In several places below we will need approximations as in 
Proposition \ref{CGapprox}, but with additional properties. In particular, we will
require that for globally hyperbolic metrics there exist approximations with 
strictly wider lightcones that are themselves globally hyperbolic.
Extending methods of \cite{Ger70}, it was shown in \cite{BM11} that global hyperbolicity is stable
in the interval topology. Consequently, if $g$ is a smooth, globally hyperbolic Lorentzian metric
then there exists some smooth globally hyperbolic metric $g'\succ g$. In \cite[Th.\ 1.2]{FS11}, the 
stability of global hyperbolicity was established for continuous cone structures. It has to be 
noted, however, that the definition of global hyperbolicity in \cite{FS11} requires stable causality
(in addition to the compactness of the causal diamonds),
which is stronger than the usual assumption of strong causality, so this result is not directly
applicable in our setting. In \cite{S14} it is proved directly that if $g$ is a continuous 
metric that is non-totally imprisoning and has the property that all causal diamonds are compact 
(as is the case for any globally hyperbolic $C^{1,1}$-metric by the proof of \cite[Lemma 14.13]{ON83}) 
then there exists a smooth metric $g'\succ g$ that has the same properties, hence in particular is
causal with compact causal diamonds and thereby globally hyperbolic by 
\cite{BS07}.
\end{remark}

\begin{Proposition}\label{CGrefined} Let $(M,g)$ be a $C^0$-spacetime 
with a smooth background Riemannian
metric $h$.
\begin{itemize}
\item[(i)] Let $\gec$, $\hat g_\eps$ as in Proposition \ref{CGapprox}. Then
for any compact subset $K\comp M$ there exists a sequence $\eps_j\searrow 0$ such that
$\hat g_{\eps_{j+1}}\prec \hat g_{\eps_{j}}$ on $K$ 
(resp.\ $\check g_{\eps_{j}}\prec \check g_{\eps_{j+1}}$ on $K$)
for all $j\in \N_0$.
\item[(ii)] If $g'$ is a continuous Lorentzian metric with $g\prec g'$ (resp.\ $g'\prec g$)
then $\hat g_\eps$ (resp.\ $\gec$) as in Proposition \ref{CGapprox} can be chosen such that 
$g\prec \hat g_\eps \prec g'$ (resp.\ $g'\prec \gec \prec g$) for all $\eps$.
\item[(iii)] There exist sequences of smooth Lorentzian metrics $\check g_j\prec g \prec \hat g_{j}$ 
($j\in \N$)
such that $d_h(\check g_j,g) + d_h(\hat g_j,g)<1/j$ and $\check g_j \prec \check g_{j+1}$ as well
as $\hat g_{j+1}\prec \hat g_{j}$ for all $j\in \N$. 
\item[(iv)] If $g$ is $C^{1,1}$ and globally hyperbolic then the $\hat g_\eps$ 
from Proposition \ref{CGapprox}, as well as the
$\hat g_j$ from (iii) can be chosen globally hyperbolic as well. 
\item[(v)] If $g$ is $C^{1,1}$ then the regularizations constructed in (i)--(iv)
can in addition be chosen such that they converge to $g$ in the $C^1$-topology and 
such that their second
derivatives are bounded, uniformly in $\eps$ (resp.\ $j$) on compact sets.
\end{itemize} 
\end{Proposition} 
\begin{proof} (i) We follow the argument of \cite[Lemma 1.5]{S14}: Pick any $\eps_0>0$. Since $g\prec \hat g_{\eps_0}$,
there exists some $\delta>0$ such that $\{X\in TM|_K\mid \|X\|_h=1,\ g(X,X)<\delta\}$ is contained in
$\{X\in TM\mid  \hat g_{\eps_0}(X,X)< 0\}$. In fact, otherwise there would exist a convergent sequence 
$X_k\to X$ in $TM|_K$ with $\|X_k\|_h=1$, $g(X_k,X_k)<1/k$, and $\hat g_{\eps_0}(X_k,X_k)\ge 0$. But then 
$g(X,X)\le 0$ and $\hat g_{\eps_0}(X,X)\ge 0$, contradicting $g\prec \hat g_{\eps_0}$. Next, we
choose $\eps_1<\min(\eps_0,\delta)$, so $d_h(g,\hat g_{\eps_1})<\delta$. Then if $X\in TM|_K$, $\|X\|_h=1$ and
$\hat g_{\eps_1}(X,X)\le 0$, we obtain $g(X,X)<  \hat g_{\eps_1}(X,X)+\delta \le \delta$, so $\hat g_{\eps_0}(X,X)<0$,
i.e., $\hat g_{\eps_1} \prec \hat g_{\eps_0}$ on $K$. The claim therefore follows by induction. Analogously one can construct
the sequence $\check g_{\eps_j}$. 

\noindent(ii) The proof of (i) shows that for any $K\comp M$ there exists some $\eps_K$ such that for all 
$\eps<\eps_K$ we have $g\prec \hat g_\eps \prec g'$ on $K$, and $d_h(g|_K,\hat g_\eps|_K)<\eps$.
Clearly all these properties are stable under shrinking $K$ or $\eps_K$. Therefore, \cite[Lemma 2.4]{KSSV}
shows that there exists a smooth map $(\eps,p)\mapsto \tilde g_\eps(p)$ such that  for each fixed $\eps$,
$\tilde g_\eps$ is a Lorentzian metric on $M$ with $g\prec \tilde g_\eps \prec g'$ and such that
$d_h(g,\tilde g_\eps)<\eps$ on $M$. Again the proof for $\gec$ is analogous.

\noindent(iii) This follows from (ii) by induction.

\noindent(iv) By Remark \ref{ghstab} there exists a smooth globally hyperbolic metric $g'\succ g$.
Constructing $\hat g_\eps$ resp.\ $\hat g_j$ as in (ii) resp.\ (iii) then automatically gives
globally hyperbolic metrics (cf.\ \cite[Sec.\ II]{BM11} ).

\noindent(v) By \cite[Lemma 2.4]{KSSV}, in the construction given in (ii) above, for any $K\comp M$, 
$\tilde g_\eps$ coincides with the original $\hat g_\eps$ on $K$ for $\eps$ sufficiently small.
Thus by (i) and (ii) from Proposition \ref{CGapprox} the $\tilde g_\eps$ (i.e., the new $\hat g_\eps$)
have the desired properties, and analogously for the new $\check g_\eps$. 
Concerning (iii), fix any atlas $\mathcal A$ of $M$ and an exhaustive sequence $K_n$ of compact
sets in $M$ with $K_n\sse K_{n+1}^\circ$ for all $n$. Then in the inductive construction
of the $\hat g_j$ we may additionally require that the $C^1$-distance of $g$ and $\hat g_j$
on $K_j$ (as measured with respect to the $C^1$-seminorms induced by the charts in $\mathcal A$) 
be less than $1/j$.
Moreover, for any $K_j$ there is some constant $C_j$ bounding  
the second derivatives of the $\hat g_\eps$ from (ii) (again w.r.t.\ the charts in $\mathcal A$)
for $\eps$ smaller than some $\eps_j$. It is therefore also possible to have the
second derivatives of $\hat g_k$ bounded by $C_j$ on $K_j$ for all $k\ge j$. 
Altogether, this gives the claimed properties for the sequence $(\hat g_j)$, and analogously for $(\check g_j)$. 
\end{proof}


\begin{Lemma}\label{approxlemma} Let $(M,g)$ be a $C^{1,1}$-spacetime and 
let $h$, $\tilde h$ be Riemannian metrics on $M$ and $TM$, respectively. 
Suppose that $\Ric(Y,Y)\ge 0$ for every Lipschitz-continuous $g$-null local vector field $Y$.
Let $K\comp M$ and let $C$, $\delta > 0$. Then there exist $\eta>0$ and $\eps_0>0$
such that for all $\eps<\eps_0$ we have: If $p\in K$ and $X\in T_pM$ is such that $\|X\|_h \le C$ 
and there exists a $g$-null vector $Y_0\in TM|_K$ with $d_{\tilde h}(X,Y_0) \le \eta$ and $\|Y_0\|_h\le C$ then
$\Ric_\eps(X,X) > -\delta$.
Here $\Ric_\eps$ is the Ricci-tensor corresponding to a metric $\hat g_\eps$ as in Proposition \ref{CGapprox}.
\end{Lemma}
\begin{proof} 
We first note that as in the proof of \cite[Lemma 3.2]{hawkingc11} it follows that we may assume 
that $M=\R^n$, $\|\,.\,\|_h = \|\,.\,\|$ is the Euclidean norm and we may replace 
$\hat g_\eps$ by $g_\eps:=g*\rho_\eps$ 
(component-wise convolution), and prove the claim for $\Ric_\eps$ calculated from $g_\eps$.
For the distance on $TM\cong \R^{2n}$ we may then simply use 
$d(X_p,Y_q) := \|p-q\|+\|X-Y\|$ (which is equivalent to the distance function induced by the
natural product metric on $T\R^n$).

Denote by $E$ the map $v\mapsto (\pi(v),\exp(v))$, defined on an open neighbourhood of the zero
section in $T\R^n$. Let $L$ be a compact neighbourhood of $K$.
Then $E$ is a homeomorphism from some open
neighbourhood $\mathcal U$ of $L\times \{0\}$ in $T\R^n$ onto an open neighbourhood 
$\mathcal V$ of $\{(q,q)\mid q\in L\}$
in $\R^n\times \R^n$ and there exists some $r>0$ such that for any $q\in L$
the set $U_r(q):=\exp_q(B_r(0))$ is a totally normal neighbourhood of $q$ and
$\bigcup_{q\in L} (U_r(q)\times U_r(q))\sse {\mathcal V}$ 
(cf.\ the proof of \cite[Th.\ 4.1]{KSS}). We may assume that $\mathcal U$ is of the form
$\{(q,v)\mid q\in L', \|v\|< a\}$ for some open $L'\supseteq L$ and some $a>0$ and
that $\overline {\mathcal U}$ is contained in the domain of $E$.  
It follows from standard ODE theory
(cf.\ \cite[Sec.\ 2]{KSS}) that 
\begin{equation}\label{geocon1}
\frac{d}{dt}(\exp^{g_\eps}_q(tv)) \to \frac{d}{dt}(\exp_q(tv)) \quad (\eps\to 0),
\end{equation}
uniformly in $v\in \R^n$ with $\|v\|\le 1$,  $t\in [0,a]$, and $q\in L$. Hence for $\eps$ small
and such $v$, $t$ and $q$ and we have 
\begin{equation}\label{geocon2}
\left\|\frac{d}{dt}(\exp_q(tv))\right\|   \le \left\|\frac{d}{dt}(\exp^{g_\eps}_q(tv))\right\| +1.
\end{equation}
Furthermore, for $\eps$ small the operator norms of $T_v\exp_q^{g_\eps}$ are bounded, 
uniformly in $\eps$, $v\in \R^n$ with $\|v\|\le a$ and $q\in L$ by some 
constant $\tilde C_1$: this follows from (7) in \cite{KSS}, noting that we may assume that 
$a$ as above is so small that this estimate is satisfied uniformly in $\eps$,
$\|v\|\le a$, and $q\in L$.
Consequently, for $\eps$ small, $q\in L$, $t\in [0,a]$ and $\|v\|\le 1$ we have
\begin{equation}\label{geocon3}
\left\|\frac{d}{dt}(\exp^{g_\eps}_q(tv))\right\|  = \left\|T_{tv}\exp^{g_\eps}_q(v)\right\| \le \tilde C_1.
\end{equation}
It follows from \eqref{geocon2}, \eqref{geocon3} that there exists some $\eps'>0$ such that for any $\eps\in (0,\eps')$,
any $q\in L$, any $v\in \R^n$ with $\|v\|\le a$  and any $t\in [0,1]$ we have
\begin{equation}\label{geocon4}
\left\|\frac{d}{dt}(\exp_q(tv))\right\| = \left\|\left.\frac{d}{ds}\right|_{s=t\|v\|}
\left(\exp_q\left(s\frac{v}{\|v\|}\right)\right)\right\|
\|v\| \le (\tilde C_1 +1)\|v\|.
\end{equation}
Set 
\begin{equation}\label{c12def}
C_1 := (\tilde C_1 +1)\sup_{p\in L}\|\Gamma(p)\|,\qquad
C_2 :=\sup_{p\in L}\|\Ric(p)\|.
\end{equation}
Given any $C>0$ and $\delta>0$, pick $\eta_1\in (0,1)$ so small that $6C_2C \eta_1<\delta/2$ and let
\begin{equation}\label{rtildef}
\tilde r := \sup\{\|E^{-1}(p,p')\| \mid p,p' \in U_r(q),\, q\in L\}.
\end{equation} 
Then $\tilde r <a$ and by compactness we may suppose that $r$ from above is so small that
$e^{C_1 \tilde r}<2$, $2C_1C\tilde r < \eta_1$, and $U_r(q)\sse L$ for all $q\in K$. 

We may then cover $K$ by finitely many such sets $U_{r}(q_1),\dots,U_{r}(q_N)$.
Then $K=\bigcup_{j=1}^N K_j$ with $K_j\comp U_j:=U_{r}(q_j)$ for each $j$. 
Set $s:=\min_{1\le j\le N}\text{dist}(K_j,\partial U_j)$ 
%
and let $0<\eta<\min(\eta_1,s/2)$.

Next, let $\rho\in {\mathcal D}(\R^n)$ be a standard mollifier, i.e., $\rho\ge 0$, 
$\text{supp}(\rho)\sse B_1(0)$ and $\int \rho(x)\,dx=1$. From (3) in  \cite{hawkingc11} we know that
\begin{equation}
R_{\eps ik} - R_{ik}*\rho_\eps \to 0 \ \text{ uniformly on compact sets}.
\end{equation}
Hence there exists some $\eps'' \in (0,\eps')$ such that for all $0<\eps<\eps''$ we have
\begin{equation}\label{rest}
\sup_{x\in K} |R_{\eps ik}(x) - R_{ik}*\rho_\eps(x)| < \frac{\delta}{2C^2}.
\end{equation}
To conclude our preparations, we set $\eps_0:=\min(\eps'',s/2)$ and consider any $\eps<\eps_0$.

Now let $p\in K$ and $X\in \R^n$ such that $\|X\| \le C$ 
and suppose there exists some $g(q)$-null vector 
$Y_0\in \R^n$ with $q\in K$, 
\begin{equation}
d(X_p,(Y_0)_q) = \|p-q\| + \|X-Y_0\| \le \eta, 
\end{equation}
and $\|Y_0\|\le C$. 
Then for some $j\in \{1,\dots,N\}$ we have $p\in K_j$, and since $\eta<s/2$ we also have
$q\in U_j$.

Since $g(q)(Y_0,Y_0)=0$, 
we may extend $Y_0$ to a Lipschitz-continuous null vector field, denoted by $Y$, on all of $U_j$ by parallelly 
transporting it radially outward from $q$.
Let $p'\in U_j$ be any point different from $q$ and let $v:=\overrightarrow{qp'}
=E^{-1}(q,p')$. Then $Y(p')=Z(1)$, where
$Z(t) = Y(\exp_q(tv))$ for all $t\in [0,1]$ and $Z$ satisfies the linear ODE
\begin{equation}\label{ode}
\frac{dZ^k}{dt} = -\Gamma_{ij}^k(\exp_q(tv))\frac{d}{dt}(\exp_q^i(tv))Z^j(t)
\end{equation}
with initial condition $Z(0)=Y(q)=Y_0$. By Gronwall's inequality it follows that 
\begin{equation}\label{zt}
\|Z(t)\| \le \|Y_0\| e^{t \|\Gamma\|_{L^\infty(U_j)}\sup_{t\in [0,1]}\|\frac{d}{dt}(\exp_q(tv))\| } \quad (t\in [0,1]).
\end{equation}
Therefore, \eqref{geocon4}, \eqref{c12def}, and \eqref{rtildef} give 
\begin{equation}\label{yp}
\|Y(p')\|\le \|Y_0\|e^{C_1\tilde r} < 2 \|Y_0\|
\end{equation}
for all $p'\in U_j$. Moreover, for all $t\in [0,1]$ we have 
\begin{equation}
\|Z(t)-Y_0\|\le t\cdot  \sup_{t\in [0,1]}\left \|\frac{dZ^k}{dt}\right\|,
\end{equation}
which, due to $\|Y_0\|\le C$, by \eqref{ode}, \eqref{zt}, and \eqref{yp} leads to 
\begin{equation}
\|Y(p')-Y_0\|\le \sup_{t\in [0,1]} \left \|\frac{dZ^k}{dt}\right\|\le C_1 C \tilde r e^{C_1\tilde r}
< 2 C_1 C \tilde r  < \eta_1.
\end{equation}
We also extend $X$ to a constant vector field on $U_j$, again denoted by $X$. 
Then $\|Y\| < 2C$ by \eqref{yp}, and 
\begin{equation}
\|X-Y\|\le \|X-Y_0\| + \|Y_0-Y\| < 2\eta_1
\end{equation} 
on $U_j$. 
It follows that, on $U_j$, we have the following inequality 
\begin{equation}
\begin{split}
|\Ric(X,X)-\Ric(Y,Y)| & = |\Ric(X-Y, X)+\Ric(X-Y,Y)|\\
&\le C_2\|X-Y\|\|X\| + C_2\|X-Y\|\|Y\| \le 6C_2C\eta_1 <\delta/2.
\end{split}
\end{equation}
Since $\Ric(Y,Y)\ge 0$, we conclude that $\Ric(X,X)>-\delta/2$ on $U_j$. 

Set
\begin{equation}
\tilde R_{ik}(x) := \left\{
\begin{array}{rl}
	R_{ik}(x) & \text{ for } x\in B_{s/2}(p)\\
	0 & \text{otherwise}.
\end{array}\right.
\end{equation}
By our assumption and the fact that $\rho\ge 0$ we then have $(\tilde R_{ik}X^iX^k)*\rho_\eps\ge -\delta/2$ on $\R^n$.
Furthermore, since $\eps<s/2$ it follows that $(R_{ik}*\rho_\eps)(p) = 
(\tilde R_{ik}*\rho_\eps)(p)$, so \eqref{rest} gives:
\begin{equation}
\begin{aligned}
|R_{\eps ik}(p)X^iX^k - ((\tilde R_{ik}X^iX^k)*\rho_\eps)(p)| &= |(R_{\eps ik}(p) - (R_{ik}*\rho_\eps)(p))X^iX^k| \\
&\le C^2 \sup_{x\in K} |R_{\eps ik}(x) - R_{ik}*\rho_\eps(x)|<\delta/2.\end{aligned}
\end{equation}
It follows that $R_{\eps ik}(p)X^iX^k>-\delta$, as claimed.
\end{proof}

\section{Proof of the main result}\label{mainproof}

Based on the approximation results of the previous section we are now ready to 
prove Theorem \ref{penrose}. As a final preliminary result we need:
\begin{Proposition} \label{eepscomp}
Let $(M,g)$ be a $C^{1,1}$-spacetime that is future null complete and suppose
that assumptions (i) and (iii) of Theorem \ref{penrose} are satisfied.
Moreover, suppose that $\hat g_\eps$ ($\eps>0$) is a net of smooth
Lorentzian metrics on $M$ 
as in Proposition \ref{CGapprox}. 
Then there exists some $\eps_0>0$ such that for all $\eps<\eps_0$ the future horismos
$E_\eps^+(\cT )$ of $\cT $ with respect to the metric $\hat g_\eps$ is relatively compact.
\end{Proposition}

\begin{proof} Let $h$ be a smooth background Riemannian metric and define
$$
\tilde T := \{v\in N(\cT )\mid v \text{ future-directed } g\text{-null and } h(v,v)=1\},
$$
where $N(\cT )$ is the $g$-normal bundle of $\cT $ and analogously 
$$
\tilde T_\eps := \{v\in N_\eps(\cT )\mid v \text{ future-directed } \hat g_\eps\text{-null and } h(v,v)=1\},
$$
where $N_\eps(\cT )$ is the $\hat g_\eps$-normal bundle of $\cT $. 
Moreover, we set (cf.\ Remark \ref{rem1.2}(b)) 
\begin{equation*}
m:=(n-2)\min_{v\in \tilde T}\conv(v) = (n-2)\min_{v\in \tilde T}g(\pi(v))(H,v) >0
\end{equation*}
and  pick $b>0$ such that $(n-2)/b<m$. 
Denote by $H_\eps$ the mean curvature vector field of $\cT $ with respect to $\hat g_\eps$, and 
similarly for $\conv_\eps$. Then $H_\eps\to H$ uniformly on $\cT $ and we claim that 
for $\eps$ sufficiently small and all $v\in \tilde T_\eps$ we have $\conv_\eps(v)>1/b$.
To see this, suppose to the contrary that there exist a sequence $\eps_k\searrow 0$ and
vectors $v_k\in \tilde T_{\eps_k}$ such that $\hat g_{\eps_k}(\pi(v_k))(H_{\eps_k},v_k)\le 1/b$
for all $k$. By compactness we may suppose without loss of generality that $v_k\to v$
as $k\to \infty$. Then $v\in \tilde T$ but $\conv(v)\le 1/b$, a contradiction.

Now we show that there exists some $\eps_0>0$ such that for all $\eps<\eps_0$
we have 
\begin{equation}\label{relcomp}
E_\eps^+(\cT ) \sse \exp^{\hat g_{\eps}}(\{sv\mid s\in [0,b],\, v\in \tilde T_{\eps}\}) \comp M.
\end{equation}
Again arguing by contradiction, suppose that there exists a sequence $\eps_j\searrow 0$ and
points $q_j\in E_{\eps_j}^+(\cT )\setminus \exp^{\hat g_{\eps_j}}(\{sv\mid s\in [0,b],\, v\in \tilde T_{\eps_j}\})$.
By \cite[Th.\ 10.51, Cor.\ 14.5]{ON83}, for each $j\in \N$ there exists a 
$\hat g_{\eps_j}$-null-geodesic $\gamma_j$ from $\cT $ to $q_j$ which is $\hat g_{\eps_j}$-normal to $\cT $ and
has no focal point before $q_j$. Let
$\gamma_j(t)=\exp^{\hat g_{\eps_j}}(t\tilde v_j)$
with $\tilde v_j\in \tilde T_{\eps_j}$. 
Let $t_j$ be such that $\gamma_j(t_j)=q_j$. Then by our indirect assumption, $t_j>b$ for all $j$.
In particular, each $\gamma_j$ is defined at least on $[0,b]$.

By compactness, we may assume that $\tilde v_j\to \tilde v$ as $j\to \infty$. Then $\tilde v\in \tilde T$, and
we set $\gamma(t):=\exp^g(t\tilde v)$. As $(M,g)$ is future-null complete,
$\gamma$ is defined on $[0,\infty)$. It now follows from standard ODE-results 
(cf.\ \cite[Sec.\ 2]{KSS})
that $\gamma_j\to \gamma$ in the $C^1$-topology on $[0,b]$. 
In particular, $\gamma_j'(t)\to \gamma'(t)$ uniformly on $[0,b]$. Pick $C>0$
and a compact set $K\Subset M$ such that $\|\gamma_j'(t)\|_h\le C$  
and $\gamma_j(t)\in K$ for all $t\in [0,b]$ and all $j\in \N$.
Then by Lemma \ref{approxlemma}, for any $\delta>0$ there exists some $j_0\in \N$ such that 
$\Ric_{\eps_j}(\gamma_j'(t),\gamma_j'(t))>-\delta$ for all $j\ge j_0$ and all $t\in [0,b]$.

Denoting by $\theta_j$ the expansion of $\gamma_j$ we have by the Raychaudhuri equation
\begin{equation}\label{deltaest}
 \frac{d(\theta_j^{-1})}{dt}\geq\frac{1}{n-2}+\frac{1}{\theta_j^2} 
 \Ric_{\hat g_{\eps_j}}({\gamma}'_j,{\gamma}'_j) > \frac{1}{n-2}-\frac{\delta}{\theta_j^2}.
\end{equation}
At this point we fix $\delta>0$ so small that
\begin{equation}\label{bc}
a:=\frac{n-2}{m} < \frac{n-2}{\alpha m} <b,
\end{equation}
where $\alpha:= 1 - (n-2)m^{-2}\delta$ and choose $j_0$ as above for this $\delta$.
For $j\ge j_0$ let $m_j:=(n-2)\min_{v\in \tilde T_{\eps_j}}\conv_{\varepsilon_j}(v)$, then $m_j\to m$ ($j\to \infty$)
and $\alpha_j:= 1 - (n-2)m_j^{-2}\delta\to \alpha$ ($j\to \infty$), so for $j$ large, \eqref{bc} implies 
\begin{equation}\label{9}
 a<\frac{n-2}{\alpha_j m_j} < b.
\end{equation}
Consequently, choosing $j$ so large that $\alpha_j>0$, the right hand side of \eqref{deltaest} is 
strictly positive at $t=0$.
Thus  $\theta_j^{-1}$ is initially strictly increasing and $\theta_j(0)=-(n-2)\conv_j(\gamma_j'(0))<-m_j<0$, so
from \eqref{deltaest} we conclude that $\theta_j^{-1}(t)\in [-m_j^{-1},0)$
on its entire domain of definition. Hence $\theta_j$ has no zero on 
$[0,b]$, whereby $\theta_j^{-1}$ exists on all of $[0,b]$. 
From this, using \eqref{deltaest}, it follows that $\theta_j^{-1}(t)
\ge f_j(t) := -m_j^{-1} + t \frac{\alpha_j}{n-2}$ 
on $[0,b]$. In particular this means that $\theta_j^{-1}$ must go to zero at or before the zero of $f_j$,
i.e., there exists some $\tau\in (0,\frac{n-2}{\alpha_j m_j})$ such that $\theta_j^{-1}(t)\to 0$ as $t\to \tau$.

But for $j$ sufficiently large \eqref{9} implies that $\theta_j^{-1}\to 0$ 
within $[0,b]$. However, since
$\gamma_j$ does not incur a focal point between $t=0$ and $t=t_j>b$,
$\theta_j$ is smooth, hence bounded, on $[0,b]$, a contradiction. 
\end{proof}

\begin{remark}\label{minass} As an inspection of the proofs of Lemma \ref{approxlemma} and Proposition
\ref{eepscomp} shows, both results remain valid for any approximating net $g_\eps$
(or sequence $g_j$) of metrics that satisfy properties (i) and (ii) from Proposition \ref{CGapprox}.
In particular, this applies to the approximations $\check g_\eps$ from the inside.
For the proof of the main result, however, it will be essential to use approximations
from the outside that themselves are globally hyperbolic.
\end{remark}

\noindent{\bf Proof of Theorem \ref{penrose}:}

Suppose, to the contrary, that $M$ is future null complete.
Proposition \ref{eepscomp} applies, in particular, to a net $\hat g_\eps$ as in 
Proposition \ref{CGrefined} (iv), approximating $g$ from the outside and such that each $\hat g_\eps$
is itself globally hyperbolic.

Fix any $\eps<\eps_0$, such that by Proposition \ref{eepscomp} $E^+_\eps(\cT )$ 
is relatively compact. Then since $\hat g_\eps$ is
globally hyperbolic, smooth causality theory (cf.\ the proof of \cite[Th.\ 14.61]{ON83})
implies that $E_{\eps}^+(\cT ) = \partial J^+_{\hat g_{\eps}}(\cT )$
is a topological hypersurface that is $\hat g_{\eps}$-achronal. 
We obtain that $E_{\eps}^+(\cT )$ is compact and since $g\prec \hat g_{\eps}$, it 
is also $g$-achronal.

As in the proof of \cite[Th.\ 14.61]{ON83} let now $X$ be a smooth $g$-timelike vector field
on $M$ and denote by $\rho: E_\eps^+(\cT )\to S$ the map that assigns to each $p\in E_\eps^+(\cT )$
the intersection of the maximal integral curve of $X$ through $p$ with $S$. Then due to the 
achronality of $E_\eps^+(\cT )$, $\rho$ is injective, so by invariance of domain it
is a homeomorphism of $E_\eps^+(\cT )$ onto an open subset of $S$. By compactness this set is
also closed in $S$. But also in the $C^{1,1}$-case, any Cauchy hypersurface is connected
(the proof of \cite[Prop.\ 14.31]{ON83} also works in this regularity).
Thus $\rho(E_\eps^+(\cT ))=S$, contradicting the fact that $S$ is non-compact.
This concludes the proof of Theorem \ref{penrose}. \hspace*{\fill}$\Box$\medskip

We also have the following analogue of \cite[Th.\ 14.61]{ON83}:
\begin{Theorem}\label{penrose_alt} Let $(M,g)$ be an $n$-dimensional $C^{1,1}$-spacetime. 
Assume that
\begin{itemize}
\item[(i)] For any Lipschitz-continuous local null vector field $X$,
$\Ric(X,X)\ge 0$.
\item[(ii)] $M$ possesses a Cauchy-hypersurface $S$.
\item[(iii)] There exists a compact spacelike achronal submanifold $\cT $
in $M$ of codimension $2$ with past-pointing timelike mean curvature vector field $H$.
\item[(iv)] $M$ is future null complete.
\end{itemize}
Then the future horismos of $\cT $, $E^+(\cT )$, is a compact Cauchy-hypersurface in $M$.
\end{Theorem}
\begin{proof} Since $(M,g)$ is globally hyperbolic, 
\cite[Prop.\ A.28]{hawkingc11} implies that the causality relation $\le$ on $M$
is closed. 
Thus since $\cT $ is compact it follows that $J^+(\cT )$ is closed. Also, by \cite[Cor.\ 3.16]{KSSV},
$J^+(\cT )^\circ=I^+(\cT )$, so $E^+(\cT )=\partial J^+(\cT )$. It is thereby the topological boundary of a
future set and the proof of \cite[Cor.\ 14.27]{ON83} carries over to the $C^{1,1}$-setting
(using \cite[Th.\ A.1, Prop.\ A.18]{hawkingc11}) to show
that $E^+(\cT )$ is a closed achronal topological hypersurface. It
remains to show that any inextendible
timelike curve intersects it. 

Suppose to the contrary that there exists some inextendible
timelike (locally Lipschitz) curve $\tilde \alpha$ that is disjoint from
$E^+(\cT )$. Then as in (the proof of) \cite[Lemma A.10]{hawkingc11} we may also
construct an inextendible timelike $C^2$-curve $\alpha$ 
that does not meet $E^+(\cT )$ (round off the breakpoints of the piecewise 
geodesic obtained in \cite[Lemma A.10]{hawkingc11} in a timelike way). 
By \cite[Ex.\ 14.11]{ON83}, since $(M,g)$ is strongly causal, $\alpha$ is an
integral curve of a timelike $C^1$-vector field $X$ on $M$. 

Next, let $\hat g_j$ be an approximating net as in  
Proposition \ref{CGrefined} (iv),(v) (to which thereby all arguments from the proof
of Theorem \ref{penrose} apply, cf.\ Remark \ref{minass}). Denote by $I^+_j(\cT )$, $J^+_j(\cT )$,
$E^+_j(\cT )$ the chronological and causal future, and the future horismos, respectively, of $\cT $
with respect to $\hat g_j$. 
Set $K:=\{sv\mid s\in [0,b],\, v\in TM|_\cT ,\, \|v\|_h=1\}\comp TM$, 
where $h$ is some complete smooth Riemannian background metric on $M$. It then follows
from the locally uniform convergence of $\exp^{\hat g_j}$ to $\exp^g$, together with 
\eqref{relcomp} that there exists some $j_0\in \N$ such that for $j\ge j_0$ we have
\begin{equation}
\partial J_j^+(\cT ) = E_j^+(\cT )\sse \exp^{\hat g_j}(K)\sse 
\overline{\{p\in M\mid \text{dist}_h(p,\exp^g(K))\le 1\}}=:L\comp M.
\end{equation}
Let the map $\rho$ from the proof of Theorem \ref{penrose} be constructed from the 
vector field $X$ from above. Then by the proof of Theorem \ref{penrose} 
we may additionally suppose that $j_0$ is such that, for 
each $j\ge j_0$, $E_j^+(\cT )$ is a compact achronal topological hypersurface
in $(M,g)$ that is homeomorphic
via $\rho$ to $S$. Therefore $\alpha$ (which is timelike for all $\hat
g_j$) intersects every $E^+_j(\cT )$ ($j\ge j_0$) precisely
once. Let $q_j$ be the intersection
point of $\alpha$ with $\partial J_{j}^+(\cT )=E^+_{j}(\cT )$. 
We now pick $t_j$ such that $q_j=\alpha(t_j)$ for all $j\in \N$. Each
$q_j$ is contained in $L$, so since $(M,g)$ is globally hyperbolic, hence
non-partially-imprisoning (as already noted in Rem.\ \ref{ghstab}, the proof of \cite[Lemma 14.13]{ON83} 
carries over verbatim to the $C^{1,1}$-case), 
it follows that $(t_j)$ is a bounded sequence in $\R$ and without loss of 
generality we may suppose that
in fact $t_j\to \tau$ for some $\tau \in \R$. Then also $q_j=\alpha(t_j)\to
q=\alpha(\tau)\in L$.

As $q_j\in \partial J_{j}^+(\cT )$ there exist $p_j\in \cT $ and $\hat g_{j}$-causal curves
$\beta_j$ from $p_j$ to $q_j$ (in fact, the $\beta_j$ are $\hat g_j$-normal 
$\hat g_j$-null geodesics). Again
without loss of generality we may assume that $p_j\to p\in \cT $.
By \cite[Th.\ 3.1]{Minguzzicurves} (or \cite[Prop.\ 2.8.1]{Chrusciel_causality}) there exists an 
accumulation curve $\beta$ of the sequence $\beta_j$ such that $\beta$ goes from $p$ to $q$.
Moreover, since $\hat g_{j+1}\prec \hat g_j$ for all $j$, each
$\beta_k$ is $\hat g_{j}$-causal for all $k\ge j$. Therefore, $\beta$
is $\hat g_{j}$-causal for each $j$. Thus by (the proof of) \cite[Prop.\ 1.5]{CG},
$\beta$ is $g$-causal and we conclude that $q=\alpha(\tau)\in J^+(\cT )$. If we had $q\in I^+(\cT )$
then for some $j_1$ we would also have $q_j\in I^+(\cT )\sse I^+_{j}(\cT )$ for all $j\ge j_1$
(using \cite[Cor.\ 3.12]{KSSV}). But this
is impossible since $q_j\in \partial J^+_{j}(\cT )=E^+_{j}(\cT )$.
Thus 
\begin{equation}
q=\alpha(\tau)\in E^+(\cT ),
\end{equation}
a contradiction to our initial assumption. We conclude that $E^+(\cT )$ is indeed a Cauchy-hypersurface in $M$.

Finally, as in the proof of Theorem \ref{penrose}, the map $\rho$ is a homeomorphism
from $E_j^+(\cT )$ onto $E^+(\cT )$ (for $j\ge j_0$), so $E^+(\cT )$ is compact.
\end{proof}

In particular, as in \cite[Cor.\ B of Th.\ 14.61]{ON83} it follows that if (i), (ii)
and (iii) from Theorem \ref{penrose_alt} hold and there exists some inextendible
causal curve that does not meet $E^+(\cT )$ then $(M,g)$ is future null incomplete.
Indeed by \cite[Lemma A.20]{hawkingc11} the existence of such a curve shows that
$E^+(\cT )$ cannot be a Cauchy-hypersurface.

\medskip\noindent
{\bf Acknowledgements.}  We would like to thank Clemens S\"amann for helpful discussions.  
This work was supported by FWF-projects P23714 and P25326.

\end{document}